\documentclass[11pt,reqno]{amsart}

\usepackage{amscd,amssymb,amsmath,amsthm}
\usepackage[arrow,matrix]{xy}
\usepackage{graphicx}
\usepackage{color}
\usepackage{cite}
\newtheorem{thm}{Theorem}
\newtheorem{defn}{Definition}
\newtheorem{lemma}{Lemma}

\newtheorem{rk}{Remark}

\numberwithin{equation}{section} 

\def\s{\sigma}

\def\s{\sigma}

\def\s{\sigma}

\def\b{\beta}

\begin{document}
\title[On free energies  of the Potts  model on the Cayley tree]{On free energies  of the Potts  model on the Cayley tree}

\author{U. A. Rozikov,  M.M. Rahmatullaev}

\address{U.\ A.\ Rozikov\\ Institute of mathematics,
29, Do'rmon Yo'li str., 100125, Tashkent, Uzbekistan.}
\email {rozikovu@yandex.ru}

\address{M.\ M.\ Rahmatullaev\\ Institute of mathematics,
29, Do'rmon Yo'li str., 100125, Tashkent, Uzbekistan.}
\email{mrahmatullaev@rambler.ru}

\begin{abstract}
For the Potts model on the Cayley tree, some explicit formulae of
the free energies and entropies (according to vector-valued
boundary conditions (BCs)) are obtained. They include
translation-invariant, periodic, Dobrushin-like BCs,
as well as those corresponding to weakly periodic Gibbs measures.
\end{abstract}
\maketitle



\section{Introduction and definitions}

On Cayley trees, not only Gibbs measures but also the free energy (and the entropy) depend on the BC.
A study of this dependence is given in \cite{GRRR}.
 It is shown there that for all previously known BCs the free energies exist.  Later, in \cite{GHRR}
 a construction  of  new Gibbs measures (called alternating Gibbs measures)
 is presented and their corresponding free energies are given. Moreover, it was proved that free energy of some alternating Gibbs measures may not exist.

The purpose of this paper is to study free energies of the Potss
model on the Cayley tree.

The Cayley tree $\Gamma^k$ (See [2]) of order $ k\geq 1 $ is an
infinite tree, i.e., a graph without cycles, from each vertex of
which exactly $ k+1 $ edges issue. Let $\Gamma^k=(V, L, i)$ ,
where $V$ is the set of vertices of $ \Gamma^k$, $L$ is the set of
edges of $ \Gamma^k$ and $i$ is the incidence function associating
each edge $l\in L$ with its endpoints $x,y\in V$. If
$i(l)=\{x,y\}$, then $x$ and $y$ are called {\it nearest
neighboring vertices}, and we write $l=\langle x,y\rangle$.

 The distance $d(x,y), x,y\in V$ on the Cayley tree is defined by
$$
d(x,y)=\min\{ d | \exists x=x_0, x_1, \dots, x_{d-1}, x_d=y\in V \ \
\mbox{such that}  \ \
 \langle x_0,x_1\rangle,\dots, \langle x_{d-1},x_d\rangle\}.$$

For the fixed $x^0\in V$ we set $ W_n=\{x\in V\ \ |\ \
d(x,x^0)=n\},$
$$ V_n=\{x\in V\ \ | \ \  d(x,x^0)\leq n\}, \ \
L_n=\{l=\langle x,y\rangle\in L \ \ |\ \  x,y\in V_n\}.$$

It is known that there exists a one-to-one correspondence between
the set  $V$ of vertices  of the Cayley tree of order $k\geq 1$
and the group $G_{k}$ of the free products of $k+1$ cyclic  groups
$\{e, a_i\}$, $i=1,...,k+1$ of the second order (i.e. $a^2_i=e$,
$a^{-1}_i=a_i$) with generators $a_1, a_2,..., a_{k+1}$.

Denote by $S(x)$ the set of {\it direct successors} of $x\in G_k$.
Let $S_1(x)$ be the set of all nearest neighboring vertices of
$x\in G_k,$ i.e. $S_1(x)=\{y\in G_k: \langle x,y\rangle\}$ and
$x_{\downarrow}$ denotes the unique element of the set
$S_1(x)\setminus S(x)$.

We consider models where spin takes values from the set $\Phi = \ \{1,
2,\dots, q \}$, $q\geq 2 $. A configuration $\s$ is defined as a function $x\in V\to\s (x) \in\Phi$; the set of
all configurations coincides with $\Omega =\Phi^{V}$.

The Hamiltonian of the Potts model has the form
\begin{equation}\label{1}
H(\sigma)=-J\sum_{\langle x,y\rangle\in L}
\delta_{\sigma(x)\sigma(y)},
\end{equation}
where $J\in R$,
$\delta_{uv}$ is the Kronecker symbol.

We identify the set $\Phi$ by the set $\{\s_1,...,\s_q\}$, where $\s_i \in R^{q-1}$ such
that
$$
\s_i \s_j=\left\{\begin{array}{ll}
-\frac{1}{q-1}, \ \ \mbox{if} \ \ i\ne j\\[2mm]
1, \ \ \mbox{if} \ \ i= j.
\end{array}\right.
$$

Then we have
\begin{equation}\label{2}
\delta_{\s(x)\s(y)}=\frac{q-1}{q}\left(\s(x)\s(y)+\frac{1}{q-1}\right).
\end{equation}

Using this formula the Hamiltonian of the Potts model can be reduced to the Hamiltonian of the Ising model with $q$ spin values:
\begin{equation}\label{1'}
H(\sigma)=-J\sum_{\langle x,y\rangle\in L}
\sigma(x)\sigma(y).
\end{equation}

We fix a basis $\{e_1,\dots,e_{q-1}\}$ on $R^{q-1}$, such that
$e_i=\s_i, i=1,2,\dots,q-1$.
It is clear that
\begin{equation}\label{3}
\sum_{i=1}^q \s_i=0.
\end{equation}

We note that if $h=(h_1, \dots, h_{q-1})$, then
$$
h \s_i=\left\{\begin{array}{ll}
\frac{q}{q-1}h_i-\frac{1}{q-1}\sum_{j=1}^{q-1}h_j, \ \ \mbox{if} \ \ i=1,\dots, q-1,\\[3mm]
-\frac{1}{q-1}\sum_{j=1}^{q-1}h_j, \ \ \mbox{if} \ \ i=q.
\end{array}\right.
$$

Define a finite-dimensional distribution of a probability measure $\mu$ in the volume $V_n$ as
 \begin{equation}\label{4}
 \mu_n(\sigma_n)=Z_n^{-1}\exp\left\{-\beta
H_n(\sigma_n)+\sum_{x\in W_n}h_x \sigma(x)\right\},
\end{equation}
where $\beta=1/T$, $T>0$ is the temperature, $ h_x\in R^{q-1}$,
$$H_n(\sigma_n)=-J\sum_{\langle x,y\rangle\in L_n}
{\sigma(x)\sigma(y)}$$  and $Z^{-1}_n$ is the normalizing factor,
i.e.
$$
Z_n=Z_n(\b,h)=\sum_{\s_n \in \Omega_n}\exp \left( -\beta
H_n(\sigma_n)+\sum_{x\in W_n}h_x\s(x)\right).
$$

The collection of vectors $h=\{h_x\in R^{q-1}, x\in V\}$
stands for (generalized) BC.

The following limit (if it exist) is called {\it free energy}
corresponding to BC $h$:
$$
E(\b,h)=-\lim_{n\to\infty }\frac{1}{\b |V_n|}\ln Z_n(\b,h).
$$

We say that probability distributions (\ref{4}) are compatible if for all
 $n\geq 1$ and $\sigma_{n-1}\in \Phi^{V_{n-1}}$ we have
 \begin{equation}\label{5}
\sum_{\s^{(n)}\in
\Phi^{W_{n}}}\mu_n(\sigma_{n-1}\vee\s^{(n)})=\mu_{n-1}(\sigma_{n-1}),
\end{equation}
where $\sigma_{n-1}\vee\s^{(n)}$ is the concatenation of the
configurations.

In this case, there exists a unique $\mu$ on $\Phi^V$ such, that for
all $n$ and $\sigma_n\in \Phi^{V_n}$ we have
$$\mu(\{\sigma|_{V_n}=\sigma_n\})=\mu_n(\sigma_n).$$

Such measure is called a {\it limiting Gibbs measure} corresponding to
Hamiltonian (\ref{1'}) and to the vector-valued function $h_x, x\in V$.

The next statement describes the condition on $h_x$ ensuring that
 $\mu_n(\sigma_n)$ are compatible.
\begin{thm}\label{t1} Measures (\ref{4}) satisfy (\ref{5}) if only if
for all $x\in V\setminus\{x^0\}$ the following equation holds:
\begin{equation}\label{6}
h_x=\sum_{y\in S(x)}F(h_y,\theta),
\end{equation} where $F: h=(h_1,
\dots,h_{q-1})\in R^{q-1}\to F(h,\theta)=(F_1,\dots,F_{q-1})\in
R^{q-1}$ is defined as
$$F_i=\ln\left({(\theta-1)e^{h_i}+\sum_{j=1}^{q-1}e^{h_j}+1\over
\theta+ \sum_{j=1}^{q-1}e^{h_j}}\right),\qquad
\theta=\exp(J\beta).$$
\end{thm}
\begin{proof} In \cite[p.106]{R} this theorem is proved for Hamiltonian (\ref{1}) (i.e. without
change (\ref{2})). Here we shall prove (\ref{6}) for Hamiltonian
(\ref{1'}), i.e. with change (\ref{2}). We show this because, a
formula obtained in this proof will be helpful to find a general
form of the free energy.

   Substituting (\ref{4}) in (\ref{5}), in view of (\ref{2}) we get
$$
\frac{Z_{n-1}}{Z_n}\sum_{\s^{(n)}}\exp\left(J\b \sum_{x\in
W_{n-1}}\left(\sum_{y\in S(x)}\s (x)\s (y)\right)+\sum_{x\in
W_{n-1}}\left(\sum_{y\in S(x)} h_y\s (y) \right)\right)=
$$
$$
\exp\left(\sum_{x\in W_{n-1}}h_x\s(x)\right), \s(x)\in \Phi.
$$

Consequently,
\begin{equation}\label{7}
\frac{Z_{n-1}}{Z_n}\prod_{x\in W_{n-1}}\prod_{y\in
S(x)}\sum_{\s(y)}\exp\left(J\b
\s(x)\s(y)+h_y\s(y)\right)=\prod_{x\in W_{n-1}}\exp(h_x\s(x)),
 \s(x)\in \Phi.
\end{equation}

Fix $x\in W_{n-1}$ and rewrite (\ref{7}) for $\s(x)=\s_i$,
$i=1,\dots,q-1$ and $\s(x)=\s_q$. Then dividing each of them by
the last one we get
\begin{equation}\label{8}
\prod_{y\in S(x)}\frac{\sum_{\s(y)}\exp((J\b
\s_i+h_y)\s(y))}{\sum_{\s(y)}\exp((J\b
\s_q+h_y)\s(y))}=\frac{\exp(h_x \s_i)}{\exp(h_x \s_q)},
i=1,2,...,q-1.
\end{equation}

Notice that $h_x\s_i-h_x\s_q=\frac{q}{q-1}h_{x, i}$, where
$h_x=(h_{x, 1}, \dots , h_{x, q-1}).$ Now we change the variables
as $\frac{q}{q-1}h_{x, i}\to h_{x, i}$ then we get the equation
(\ref{6}). For a proof of the inverse statement see
\cite[p.107]{R}.
\end{proof}

 Let $G_k/G_k^*=\{H_1,...,H_r\}$ be the quotient group, where $G_k^*$ is a normal subgroup of index $r\geq 1$.

\begin{defn} A set of vectors  $h=\{h_x,\, x\in G_k\}$
is said to be $G_k^*$ periodic, if $h_{yx}=h_x$ for any $ x\in
G_k$ and $y\in G^*_k.$
\end{defn}

\begin{defn}  A set of vectors  $h=\{h_x,\, x\in G_k\}$
is said to be  $G_k^*$ weakly periodic, if $h_{x}=h_{ij}$ for
$x\in H_i$ and $x_\downarrow\in H_j$ for any $x\in G_k$.
\end{defn}

\begin{defn} A measure $\mu$ is said to be
$G_k^*$-periodic (weakly periodic), if it corresponds to the
$G_k^*$-periodic (weakly periodic) set of vectors $h$. The $G_k-$
periodic measure is said to be translation-invariant.
\end{defn}

In this paper we compute free energies which correspond to
translation-invariant, periodic, weakly periodic and some non-periodic BCs (Gibbs measures).

 \section{Formula of free energy}

 The following theorem gives a general form of the free energy.

\begin{thm}\label{t2} For BCs satisfying (\ref{5}), the free energy is given by the formula
\begin{equation}\label{E}
E(\b,h)=-\lim_{n\to\infty }\frac{1}{|V_n|}\sum_{x\in V_n}a(x),
\end{equation}
where
\begin{equation}\label{a}
a(x)=\frac{1}{q\b}\sum_{i=1}^q\ln \left(\sum_{u=1}^q \exp\{(J\b
\s_i+h_x)\s_u\} \right).
\end{equation}
\end{thm}
\begin{proof} In view of (\ref{8}) we have
$$
\prod_{y\in S(x)}\sum_{u=1}^q\exp(J\b \s_i \s_u+\s_u
h_y)=b(x)\exp(\s_i h_x), \ \ i=1,\dots,q.
$$
Multiplying all such equalities and using (\ref{3}) we obtain
$$
b^q(x)=\prod_{i=1}^q\prod_{y\in S(x)}\sum_{u=1}^q \exp(J\b \s_i
\s_u+\s_u h_y)=\prod_{y\in
S(x)}\prod_{i=1}^q\sum_{u=1}^q \exp(J\b \s_i + h_y)\s_u,
$$
hence
$$
b(x)=\prod_{y\in S(x)}\prod_{i=1}^q \left(\sum_{u=1}^q \exp(J\b
\s_i + h_y)\s_u\right)^{1/q}.
$$
Let $A_n=\prod_{x\in W_n}b(x).$ It is clear that
$Z_n=A_{n-1}Z_{n-1}.$ Consequently
$$Z_n=\prod_{x\in W_{n-1}}b(x)Z_{n-1}=\prod_{x\in
W_{n-1}}b(x)\prod_{x\in W_{n-2}}b(x)Z_{n-2}=...=\prod_{x\in
V_{n-1}}b(x),$$ and
$$
\ln Z_n=\sum_{x\in V_{n-1}}\ln b(x).
$$
Hence we obtain that
$$
a(x)=\frac{1}{q\b}\sum_{i=1}^q \ln\left(\sum_{u=1}^q \exp\{(J\b
\s_i+h_x)\s_u\} \right).
$$
\end{proof}

\section{Free energies corresponding to some BCs}

\subsection{Translation-invariant case}

Consider translation-invariant set of vectors $h_x$, i.e.
$h_x=h=(h_1,h_2,...,h_{q-1})\in R^{q-1}, \forall x\in G_k.$ Then
from (\ref{6}) we obtain
\begin{equation}\label{11}
h_i=k\ln\left({(\theta-1)e^{h_i}+\sum_{j=1}^{q-1}e^{h_j}+1\over \theta+ \sum_{j=1}^{q-1}e^{h_j}}\right),\ \ i=1,\dots,q-1.
\end{equation}
Denoting $z_i=\exp(h_i), i=1,\dots,q-1$, we get from (\ref{11})
\begin{equation}\label{12}
z_i=\left({(\theta-1)z_{i}+\sum_{j=1}^{q-1}z_{j}+1\over \theta+
\sum_{j=1}^{q-1}z_{j}}\right)^k, i=1,2,...,q-1.
\end{equation}

For $k=2$ we denote
\begin{equation}\label{13}
x_1(m)={\theta-1-\sqrt{(\theta-1)^2-4m(q-m)}\over 2m}, \ \
x_2(m)={\theta-1+\sqrt{(\theta-1)^2-4m(q-m)}\over 2m},
\end{equation}
where
$$
\theta\geq \theta_m=1+2\sqrt{m(q-m)}, \ \ m=1,\dots,q-1.$$

It is easy to see that
$$
\theta_m=\theta_{q-m} \ \ \mbox{and} \ \
\theta_{1}<\theta_2<\dots<\theta_{[{q\over 2}]-1}<\theta_{[{q\over
2}]}\leq q+1.$$

 Let $k=2$, $J>0$, then the following statements are known (see \cite{KRK}).
\begin{itemize}
\item[1.]
If $\theta<\theta_1$, then the system of equations  (\ref{11}) has
a unique solution $h_0=(0,0,\dots,0)$;

\item[2.]
If $\theta_{m}<\theta<\theta_{m+1}$ for some $m=1,\dots,[{q\over
2}]-1$, then the system of equations (\ref{11}) has solutions
$$h_0=(0,0,\dots,0), \ \ h_{1i}(s), \ \ h_{2i}(s), \ \ i=1,\dots, {q-1\choose s}, $$
$$h'_{1i}(q-s), \ \ h'_{2i}(q-s), \ \ i=1,\dots, {q-1\choose q-s}, \ \ s=1,2,\dots,m,$$
where $h_{ji}(s)$, (resp. $h'_{ji}(q-s)$)\, $j=1,2$ is a vector
with $s$ (resp. $q-s$) coordinates equal to $2\ln x_j(s)$ (resp.
$2\ln x_j(q-s)$) and remaining $q-s-1$ (resp. $s-1$) coordinates
equal to 0. The number of such solutions is equal to
$$1+2\sum_{s=1}^m{q\choose s};$$

\item[3.] If $\theta_{[{q\over 2}]}<\theta\ne q+1$, then there are $2^q-1$ solutions to (\ref{11});

\item[4] If $\theta=q+1$ then the
number of solutions is as follows
$$\left\{\begin{array}{ll}
2^{q-1}, \ \ \mbox{if} \ \ q \ \ \mbox{is odd}\\[2mm]
2^{q-1}-{q-1\choose q/2}, \ \ \mbox{if} \ \ q \ \ \mbox{is even};
\end{array}\right.$$

\item[5.] If $\theta=\theta_m$, $m=1,\dots,[{q\over 2}]$, \,($\theta_{[{q\over 2}]}\ne q+1$) then  $h_{1i}(m)=h_{2i}(m)$. The number of solutions is equal to
$$1+{q\choose m}+2\sum_{s=1}^{m-1}{q\choose s}.$$\end{itemize}

Thus any solution of (\ref{11}) has the form
\begin{equation}\label{h}
h=(\underbrace{h_*, h_*,\dots,h_*}_m, 0, 0,\dots, 0), \ \ m\geq 0
\end{equation}
up to a permutation of coordinates.

In this subsection we shall calculate free energies and entropy
$S(\b,h)=-\frac{dE(\b, h)}{d T}$ for the set of
translation-invariant vectors $h_x=h$, with $h$ given by (\ref{h}).

{\it Case} $m=0$. In this case $h=h_0=(0,0,...,0)\in R^{q-1}$. From (\ref{E}) we have
$$
E_{TI}(\b,h_0)= -a(x)=-\frac{1}{q\b}\sum_{i=1}^q \ln\left(
\sum_{u=1}^q \exp(J\b\s_i\s_u)\right)=$$
$$-\frac{1}{q\b}\sum_{i=1}^q \ln\left(
\exp(J\b)+(q-1)\exp\left(\frac{J\b}{1-q}\right)\right)=
$$
$$
-\frac{1}{q\b}q\ln\left(
\exp(J\b)+(q-1)\exp\left(\frac{J\b}{1-q}\right)\right)=-J-\frac{1}{\b}\ln
\left(1+(q-1)\exp\left(\frac{Jq\b}{1-q}\right)\right).
$$

The corresponding entropy has the form
$$
S_{TI}(\b, h_0)=-\frac{dE_{TI}(\b, h_0)}{d
T}=\ln\left(1+(q-1)\exp\left(\frac{Jq\b}{1-q}\right)\right)+\frac{Jq\b
\exp\left(\frac{Jq\b}{1-q}\right)}{1+(q-1)\exp\left(\frac{Jq\b}{1-q}\right)}.
$$

{\it Case}  $m\neq 0$. Using (\ref{h}) we shall calculate the free
energy:
$$
E_{TI}(\b, m, h_x)=-\lim_{n\to\infty }\frac{1}{|V_n|}\sum_{x\in
V_n}a(x)=-a(x)=-\frac{1}{q\b}\sum_{i=1}^q \ln\left(\sum_{u=1}^q
\exp\{(J\b \s_i+h_x)\s_u\} \right)=$$
$$-\frac{q-m}{q\b}\ln\left(m\cdot e^{\left(-\frac{J\b}{q-1}+\frac{q-m}{q-1}h_*\right)}+e^{\left(J\b-
\frac{m}{q-1} h_*\right)}+(q-m-1)\cdot e^{\left(-\frac{J\b}{q-1}-\frac{m}{q-1}h_*\right)}\right)-$$
\begin{equation}\label{14}
\frac{m}{q\b}\ln\left((m-1)\cdot e^{\left(-\frac{J\b}{q-1}+\frac{q-m}{q-1}h_*\right)}+
e^{\left(J\b+\frac{q-m}{q-1} h_*\right)}+(q-m)\cdot e^{\left(-\frac{J\b}{q-1}-\frac{m}{q-1}h_*\right)}\right).\end{equation}

Taking into account $h_*=2\ln x_j(m)$, $j=1,2$ (where the $x_j(m)$
are defined by (\ref{13})) we calculate the corresponding entropy:
$$
S_{TI}(\b, m, h_x)=-\frac{dE_{TI}(\b, m, h_x)}{d T}=-E_{TI}(\b, m,
h_x)\b+$$
$$\frac{J\b (q-m)}{q\left(m\cdot
e^{\left(-\frac{J\b}{q-1}+\frac{q-m}{q-1}h_*\right)}+e^{\left(J\b-
\frac{m}{q-1} h_*\right)}+(q-m-1)\cdot
e^{\left(-\frac{J\b}{q-1}-\frac{m}{q-1}h_*\right)}\right)}\times$$
$$
\left[m(1+(q-m)A)
e^{\left(-\frac{J\b}{q-1}+\frac{q-m}{q-1}h_*\right)}+(1-q-mA)e^{\left(J\b-
\frac{m}{q-1} h_*\right)}+\right.$$
$$\left.(q-m-1)(1-mA)
e^{\left(-\frac{J\b}{q-1}-\frac{m}{q-1}h_*\right)}\right]+
$$
$$
 \frac{J\b m}{q\left((m-1)\cdot
e^{\left(-\frac{J\b}{q-1}+\frac{q-m}{q-1}h_*\right)}+
e^{\left(J\b+\frac{q-m}{q-1} h_*\right)}+(q-m)\cdot
e^{\left(-\frac{J\b}{q-1}-\frac{m}{q-1}h_*\right)}\right)}\times$$
$$\left[(m-1)(1+(q-m)A)
e^{\left(-\frac{J\b}{q-1}+\frac{q-m}{q-1}h_*\right)}+((m-q)A-1)e^{\left(J\b-
\frac{m}{q-1} h_*\right)}+\right.$$
\begin{equation}\label{en}
\left.(q-m)(1-mA)
e^{\left(-\frac{J\b}{q-1}-\frac{m}{q-1}h_*\right)}\right],
\end{equation}
where $A=\frac{2e^{J\b}}{(q-1)(e^{J\b}-2me^{\frac{h_*}{2}}-1)}$.

\begin{rk} We note that the entropy for $ 2\ln x_1(m)$ and $ 2\ln x_2(m)$
can be calculated by the formula (\ref{en}) replacing $h_*$ by $ 2\ln x_1(m)$ and $ 2\ln x_2(m)$ respectively.
\end{rk}

\begin{rk} Notice that under any permutations of the coordinates
of the vector $h_x$ the free energy and entropy
do not change.
\end{rk}

 In Fig.1 graphs of the free energies (\ref{14}) are shown.

\begin{center}
\includegraphics[width=10cm]{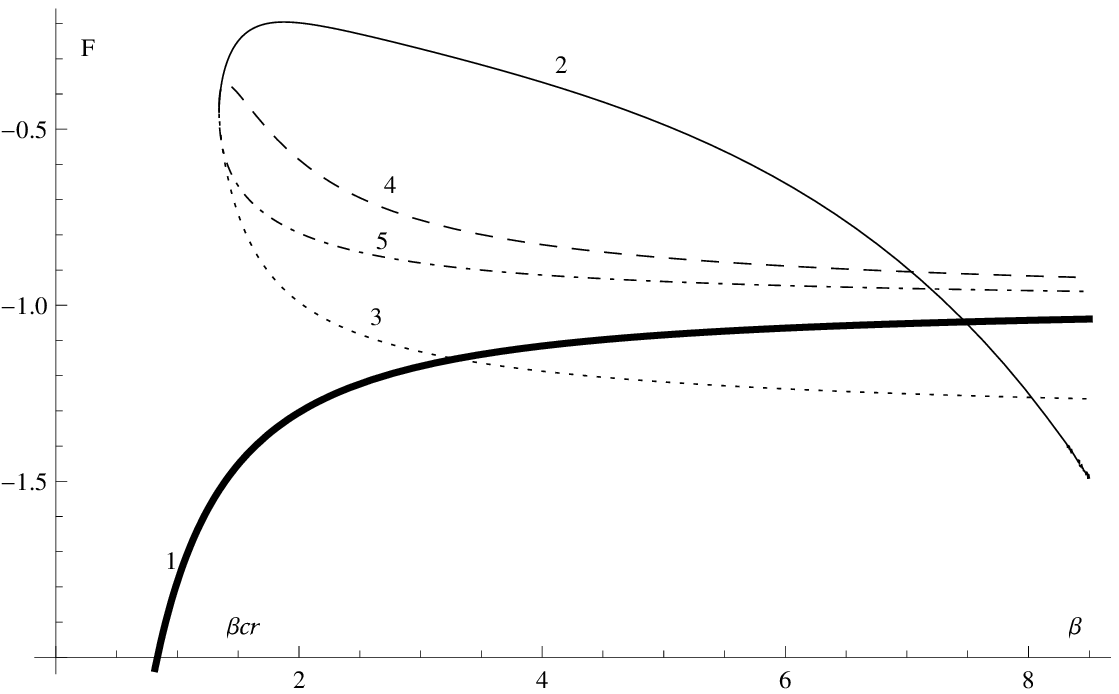}
\end{center}
{\footnotesize \noindent
 Fig.~1.
Case $q=3$. The free energy $F(\b,h_0)$ the bold solid line (line 1); the free
energy $F(\b, (2\ln x_1(1),0))$
solid line (line 2); the free energy $F(\b, (2\ln
x_2(1),0))$ the dotted line (line 3);\\ the free
energy $F(\b, (2\ln x_1(2),2\ln x_1(2)))$ the dashed line
(line 4); the free energy $F(\b, (2\ln x_2(2),2\ln
x_2(2)))$ the dotted-dashed line (line 5), where $x_i(m), i=1,2,
m=1,2$ defined in (\ref{13}).}\\

\subsection{A non-translation-invariant BCs} In this subsection we consider non-translation-invariant BCs (Gibbs measures)
 constructed by N.Ganikhadjaev in \cite{G}. Here one considers the
half tree. Namely the root $x^0$ has $k$ nearest neighbors.
Consider an infinite path $\pi=\{x^0=x_0<x_1<\dots\}$  (the
notation $x<y$ meaning that  pathes from the root to $y$ go
through $x$).
Take two different solutions $h_*^1$ and $h_*^2$  of (\ref{11}) (having the form (\ref{h})).  Associate  to this  path a collection $h^\pi$ of
vectors given by the condition
\begin{equation}\label{b3.1}
h_x^\pi=\left\{\begin{array}{ll}
h^1_*, \ \ \mbox{if} \ \ x\prec x_n, \, x\in W_n,\\[2mm]
h^2_*, \ \ \mbox{if} \ \ x_n\preceq x, \, x\in W_n,
\end{array}\right.
\end{equation}
$n=1,2,\dots$ where $x\prec x_n$ (resp. $x_n\prec x$) means that
$x$ is on the left (resp. right) from the path $\pi$.

For a given infinite path $\pi$ we put
$${\mathcal W}_n^\pi=\{x\in W_n:
x\prec x_n\},$$ where $n=1,2,\dots$.
\begin{lemma} For any $\pi$ there exists the limit
\begin{equation}
\lim_{n\rightarrow \infty}\frac{|{\mathcal
W}_n^\pi|}{|W_n|}=a^\pi.
\end{equation}
\end{lemma}
\begin{proof} Put
$$\Delta_1^\pi={\mathcal W}_1^\pi;\quad \Delta_n^\pi={\mathcal W}_n^\pi\cap S(x_{n-1}),$$
where $n\geq 2$.

It is easy to see that
$$|{\mathcal W}_n^\pi|=\sum_{i=1}^n|\Delta_i^\pi|\cdot |W_{n-i}|$$
and
$0\leq|\Delta_i^\pi|\leq k-1$, for any $i, \pi.$

Let $a_n^\pi=\frac{|{\mathcal W}_n^\pi |}{|W_n|}.$ Then using
$|W_n|=k^{n-1}(k+1)$ we get
$$
a_{n+1}^\pi-a_n^\pi=\frac{|{\mathcal
W}_{n+1}^\pi|}{|W_{n+1}|}-\frac{|{\mathcal W}_n^\pi
|}{|W_n|}=\frac{\Delta_{n+1}^\pi }{k^n(k+1)}\geq 0,
$$
i.e. the sequence $a_n^\pi$ is a monotone non-decreasing. It is
easy to see that $a_n^\pi < 1$.

Consequently $\lim_{n\rightarrow \infty}\frac{|{\mathcal
W}_n^\pi|}{|W_n|}=a^\pi$ exists and is finite.
\end{proof}

Denote $V_n^l=\cup_{i=1}^n {\mathcal W}_n^\pi.$ By the
Stolz-Ces\'aro theorem (see e.g. \cite{K}) we have

\begin{equation}\label{b}\lim_{n\to\infty
}\frac{|V_n^l|}{|V_n|}=\lim_{n\rightarrow \infty}\frac{|{\mathcal
W}_n^\pi|}{|W_n|}=a^\pi.
\end{equation}

 Now we calculate the free
energy for the set of vectors $h_x^\pi$. In view of (\ref{b}) we
get
$$
E_G(\b, m, h_x^\pi)=-\lim_{n\to\infty }\frac{1}{|V_n|}\sum_{x\in
V_n}a(x)=-\frac{1}{q\b}\lim_{n\to\infty
}\frac{|V_n^l|}{|V_n|}\cdot\sum_{i=1}^q \ln\left(\sum_{u=1}^q
\exp\{(J\b \s_i+h_*^1)\s_u\} \right)-$$
$$-\frac{1}{q\b}\lim_{n\to\infty
}\left(1-\frac{|V_n^l|}{|V_n|}\right)\cdot \sum_{i=1}^q
\ln\left(\sum_{u=1}^q \exp\{(J\b \s_i+h_*^2)\s_u\} \right)=$$
\begin{equation}\label{EG}
a^\pi E_{TI}(\b,m,h_*^1)+(1-a^\pi) E_{TI}(\b,m,h_*^2),
\end{equation} where $E_{TI}(\b,m,h)$ is defined in (\ref{14}).

Using (\ref{en}) and (\ref{EG}) we get the following formula for
the corresponding entropy:
$$
S_G(\b, m, h_x^\pi)=a^\pi S_{TI}(\b, m, h^1_x)+(1-a^\pi)S_{TI}(\b,
m, h^2_x).
$$

\subsection{Periodic BCs.}

In this subsection we consider periodic BCs and will be
calculating free energies for them.

We consider the case $q=3$. It is known (see \cite{RK}) that there are only
 $G^{(2)}_k$-periodic Gibbs measures, where $G^{(2)}_k$ is the set of all words of even lengths.
 The corresponding set of vectors $h=\{h_x\in R^{q-1}: \, x\in
G_k\}$ has the form
$$h_x=\left\{%
\begin{array}{ll}
    {\bf h}^1, \ \ \mbox{if} \ \  x\in G_k^{(2)}, \\[2mm]
    {\bf h}^2, \ \ \mbox{if} \ \  x\in G_k\setminus G_k^{(2)},
\end{array}%
\right. $$
 where ${\bf h}^1=(h_1^1,h_2^1)$, ${\bf h}^2=(h_1^2,h_2^2)$.

In view of (\ref{6}) we have
\begin{equation}\label{ph}
\left\{%
\begin{array}{ll}
    h_{1}^1=k\ln\left({\theta\exp(h_1^2) + \exp(h_2^2)+1\over \exp(h_1^2) + \exp(h_2^2)+\theta}\right)\\[3 mm]
    h_{2}^1=k\ln\left({\exp(h_1^2) +\theta \exp(h_2^2)+1\over \exp(h_1^2) + \exp(h_2^2)+\theta}\right) \\[3 mm]
    h_{1}^2=k\ln\left({\theta\exp(h_1^1) + \exp(h_2^1)+1\over \exp(h_1^1) + \exp(h_2^1)+\theta}\right) \\[3 mm]
    h_{2}^2=k\ln\left({\exp(h_1^1) + \theta\exp(h_2^1)+1\over \exp(h_1^1) + \exp(h_2^1)+\theta}\right). \\
\end{array}%
\right.
\end{equation}

For $k=3$,\, $J<0$ it is known (see \cite{Hakimov}) that if
$0<\theta<{1\over 4}$ then the system (\ref{ph}) has at least two
solutions of the form: $h=(h^1,h^1,h^2,h^2)$. Now for such set of
periodic vectors $h_x$, we shall describe free energies. From
(\ref{a}) we have
 $$a(x)=\left\{%
\begin{array}{ll}
    d(\textbf{h}^1),$ if $ x\in G_k^{(2)} \\
    d(\textbf{h}^2),$ if $ x\in G_k\setminus G_k^{(2)}, \\
\end{array}%
\right.$$  where
\begin{equation}\label{15}
 d(\textbf{h})=\frac{1}{q\b}\sum_{i=1}^q \left(\sum_{u=1}^q \exp\{(J\b
\s_i+\textbf{h})\s_u\} \right).
\end{equation}

Denote
$$V_{even,n}=\{x\in V_n:
x\in G_k^{(2)}\}, \ \  V_{odd,n}=\{x\in V_n: x\in G_k\setminus
G_k^{(2)}\}.$$ It is easy to check that for $n=2p$,
$$
 |V_{even,
2p}|=\frac{k^{2p+1}-1}{k-1},\ \  |V_{odd,
2p}|=\frac{k^{2p}-1}{k-1},
$$
for $n=2p+1,$
$$
  |V_{even, 2p+1}|=\frac{k^{2p}-1}{k-1}, \ \ |V_{odd,
2p+1}|=\frac{k^{2p+2}-1}{k-1},\ \
$$ also we have $$|V_n|=\frac{(k+1)k^n-2}{(k-1)}.$$

Using these formulas we calculate the free energy:
$$
E_{per}(\b, h_x)=-\lim_{n\to \infty}\frac{|V_{even,
n}|d(\textbf{h}^1)+|V_{odd, n}|d(\textbf{h}^2)}{|V_n|}=$$
$$-\lim_{n\to \infty}\left\{%
\begin{array}{ll}
    \frac{|V_{even,
2p}|d(\textbf{h}^1)+|V_{odd, 2p}|d(\textbf{h}^2)}{|V_{2p}|} \ \ \mbox{if} \ \  n=2p \\[3mm]
   \frac{|V_{even,
2p+1}|d(\textbf{h}^1)+|V_{odd, 2p+1}|d(\textbf{h}^2)}{|V_{2p+1}|} \ \ \mbox{if} \ \  n=2p+1,  \\
\end{array}%
\right.=
$$
$$ -\frac{1}{k+1}\left\{%
\begin{array}{ll}
   kd(\textbf{h}^1)+d(\textbf{h}^2) \ \ \mbox{if} \ \  n=2p, \, p\rightarrow \infty\\[3mm]
   d(\textbf{h}^1)+kd(\textbf{h}^2)  \ \ \mbox{if} \ \ n=2p+1.  \\
\end{array}%
\right.
$$

From this equality it follows that if $d(\textbf{h}^1)\neq
d(\textbf{h}^2)$, then for the periodic BCs a free energy does not
exist. For $k=q=3$ and fixed $\theta=\frac{1}{5}$, using a
computer analysis one can see that $d(\textbf{h}^1)\neq
d(\textbf{h}^2)$.

\begin{rk} It is known (see \cite{GRRR}) that for periodic
Gibbs measures of the Ising model on Cayley trees the free
energies exist. But for the Potts model we proved that free energy
of periodic Gibbs measures may not exist.
\end{rk}
\subsection{Weakly periodic BCs.}

Construct a weakly periodic BC. For $A\subset \{1,2,...,k+1\}$ we
consider $H_A=\{x\in G_k: \sum_{j\in A}w_j(x)-$even$\},$ where
$w_j(x)$ is the number of $a_j$ in a word $x$,
$G_k/H_A=\{H_A,G_k\setminus H_A\}$ is a quotient group. For
simplicity, we set $H_0=H_A,$ $ H_1=G_k\setminus H_A$. The $H_A$ -
weakly periodic sets of vectors $h=\{h_x\in R^{q-1}: \, x\in
G_k\}$ have the following form
\begin{equation}\label{16}
h_x=\left\{%
\begin{array}{ll}
    h_1, & \textrm{if} \ \ {x_{\downarrow} \in H_0}, \ x \in H_0 \\
    h_2, & \textrm{if} \ \ {x_{\downarrow} \in H_0}, \ x \in H_1 \\
    h_3, & \textrm{if} \ \ {x_{\downarrow} \in H_1}, \ x \in H_0 \\
    h_4, & \textrm{if} \ \ {x_{\downarrow} \in H_1}, \ x \in H_1. \\
\end{array}%
\right.
\end{equation}
Here $h_i=(h_{i1},h_{i2},...,h_{iq-1}),$ $i=1,2,3,4.$ By (\ref{6}), we have
\begin{equation}\label{17}
\left\{%
\begin{array}{ll}
    h_{1}=(k-|A|)F(h_{1},\theta)+|A|F(h_{2},\theta) \\
    h_{2}=(|A|-1)F(h_{3},\theta)+(k+1-|A|)F(h_{4},\theta) \\
    h_{3}=(|A|-1)F(h_{2},\theta)+(k+1-|A|)F(h_{1},\theta) \\
    h_{4}=(k-|A|)F(h_{4},\theta)+|A|F(h_{3},\theta).\\
\end{array}%
\right.
\end{equation}

In \cite{RM} it was shown that for $|A|=k$, $k\geq 6$ the system
of equations (\ref{17}) has at least two (not
translation-invariant) solutions, which generate sets of vectors
$h_x$ of the form of (\ref{16}), where all coordinates of vectors
$h_i, i=1,2,3,4$ are equal and $h_i\neq h_j$ for $i\neq j$. Now
for such weakly periodic sets of vectors $h_x$, we calculate the
corresponding free energy.

We introduce the following
$$\mathcal A_n=\left|\left\{\langle x, y\rangle\in L_n:\, x\in H_0,\,
y=x_{\downarrow}\in H_0\right\}\right|,$$
$$
\mathcal B_n=\left|\left\{\langle x, y\rangle\in L_n:\, x\in
H_0,\, y=x_{\downarrow}\in H_1\right\}\right|,
$$
$$\mathcal C_n=\left|\left\{\langle x, y\rangle\in L_n:\, x\in H_1,\,
y=x_{\downarrow}\in H_0\right\}\right|,$$
$$\mathcal D_n=\left|\left\{\langle x, y\rangle\in L_n:\, x\in H_1,\,
y=x_{\downarrow}\in H_1\right\}\right|,
$$
where $L_n$ is a set of edges in $V_n$.

It is known from \cite{GRRR} that

$$\lim_{n\to\infty} {\mathcal A_{n}\over |V_n|}=\lim_{n\to\infty} {\mathcal D_{n}\over |V_n|}={{1}\over 2(k+1)},$$
$$\lim_{n\to\infty} {\mathcal B_{n}\over |V_n|}=\lim_{n\to\infty} {\mathcal C_{n}\over |V_n|}={k\over 2(k+1)}.$$

Using these formulas we calculate the free energy:
$$E_{WP}(\b, q, h_x)=-\lim_{n\to\infty }\frac{1}{|V_n|}(\mathcal A_n
d(h_1)+\mathcal B_n d(h_2)+\mathcal C_n d(h_3)+\mathcal D_n
d(h_4))=$$
$$\frac{1}{2(k+1)}\left(E_{TI}(\b,q-1,h_1)+k
E_{TI}(\b,q-1,h_2)+kE_{TI}(\b,q-1,h_3)+E_{TI}(\b,q-1,h_4)\right),$$
where
 $E_{TI}(\b,m,h_x)$ is defined in (\ref{14}).
Now the entropy is
$$ S_{WP}(\b, q,
h_x)=\frac{1}{2(k+1)}\left(S_{TI}(\b,q-1,h_1)+k
S_{TI}(\b,q-1,h_2)+\right.$$$$\left.
kS_{TI}(\b,q-1,h_3)+S_{TI}(\b,q-1,h_4)\right),$$ where
 $S_{TI}(\b,m,h_x)$ is defined in (\ref{en}).

\end{document}